\documentclass[3p,11pt]{elsarticle}

\usepackage{hyperref}
\usepackage{latexsym}
\usepackage{xspace}
\usepackage{amssymb}
\usepackage{amsmath}
\usepackage{stmaryrd}
\usepackage[all]{xy}
\usepackage{calc}
\usepackage{tikz}
\usepackage{color}
\usepackage{enumerate}
\usepackage[ruled,vlined,linesnumbered]{algorithm2e}

\newtheorem{theorem}{Theorem}[section]

\newtheorem{proposition}[theorem]{Proposition}
\newtheorem{corollary}[theorem]{Corollary}

\newtheorem{Definition}[theorem]{Definition}
\newtheorem{Example}[theorem]{Example}
\newtheorem{Remark}[theorem]{Remark}

\newenvironment{definition}{\begin{Definition}\begin{em}}{\end{em}\end{Definition}}
\newenvironment{example}{\begin{Example}\begin{em}}{\end{em}\end{Example}}
\newenvironment{remark}{\begin{Remark}\begin{em}}{\end{em}\end{Remark}}
\newproof{proof}{Proof}

\SetKwInput{Input}{Input}
\SetKwInput{Output}{Output}
\SetKwInput{Purpose}{Purpose}
\SetKw{Break}{break}
\SetKw{Continue}{continue}
\SetKw{Return}{return}
\SetKw{New}{new\,}

\def\eqref#1{(\ref{#1})}

\def\tuple#1{\langle#1\rangle}

\newcommand{\E}{\exists}

\newcommand{\myend}{\mbox{}\hfill$\Box$}

\newcommand{\comment}[1]{}

\newcommand{\fand}{\varotimes}

\newcommand{\fto}{\Rightarrow}
\newcommand{\fequiv}{\Leftrightarrow}

\newcommand{\FTSs}{FTSs\xspace}
\newcommand{\NFTS}{NFTS\xspace}
\newcommand{\NFTSs}{NFTSs\xspace}
\newcommand{\NFLTS}{NFLTS\xspace}
\newcommand{\NFLTSs}{NFLTSs\xspace}
\newcommand{\FLG}{FLG\xspace}
\newcommand{\FLGs}{FLGs\xspace}
\newcommand{\FLTS}{FLTS\xspace}
\newcommand{\FLTSs}{FLTSs\xspace}
\newcommand{\fALC}{$\mathit{f}\mathcal{ALC}$\xspace}

\newcommand{\mS}{\mathcal{S}}
\newcommand{\mF}{\mathcal{F}}

\newcommand{\SA}{A}
\newcommand{\sS}{\Sigma}
\newcommand{\SV}{\Sigma_V}
\newcommand{\SE}{\Sigma_E}

\newcommand{\support}{\mathit{support}}
\newcommand{\size}{\mathit{size}}
\newcommand{\elements}{\mathit{elements}}
\newcommand{\degree}{\mathit{degree}}
\newcommand{\subblocks}{\mathit{subblocks}}
\newcommand{\anyElem}{\mathit{anyElement}}
\newcommand{\allElems}{\mathit{allElements}}

\newcommand{\CompCBt}{\mbox{$\mathsf{ComputeBisimulationEfficiently}$}\xspace}
\newcommand{\CompCPfNFTS}{\mbox{$\mathsf{ComputeCrispPartitionNFTS}$}\xspace}

\newcommand{\CompFPt}{\mbox{$\mathsf{ComputeFuzzyPartitionEfficiently}$}\xspace}
\newcommand{\CompFPfNFTS}{\mbox{$\mathsf{ComputeFuzzyPartitionNFTS}$}\xspace}
\newcommand{\bbP}{\mathbb{P}}
\newcommand{\bbB}{\mathbb{B}}
\newcommand{\rs}{\mathit{result}}

\newcommand{\deltatt}{\delta_{_\circ}}
\newcommand{\InitComplexity}{$O(|S| + \size(\delta))$}

\journal{arXiv}

\begin{document}
\sloppy
	
\begin{frontmatter}		
		
\title{Efficient algorithms for computing bisimulations for nondeterministic fuzzy transition systems}

\author{Linh Anh Nguyen}
\ead{nguyen@mimuw.edu.pl}
				
\address{Institute of Informatics, University of Warsaw, Banacha 2, 02-097 Warsaw, Poland}
%\address{Faculty of Information Technology, Nguyen Tat Thanh University, Ho Chi Minh City, Vietnam}

\begin{abstract}
Fuzzy transition systems offer a robust framework for modeling and analyzing systems with inherent uncertainties and imprecision, which are prevalent in real-world scenarios. As their extension, nondeterministic fuzzy transition systems (\NFTSs) have been studied in a considerable number of works. Wu {\em et al.} (2018) provided an algorithm for computing the greatest crisp bisimulation of a finite \NFTS $\mS = \tuple{S, \SA, \delta}$, with a time complexity of order $O(|S|^4 \cdot |\delta|^2)$ under the assumption that $|\delta| \geq |S|$. Qiao {\em et al.} (2023) provided an algorithm for computing the greatest fuzzy bisimulation of a finite \NFTS $\mS$ under the G\"odel semantics, with a time complexity of order $O(|S|^4 \cdot |\delta|^2 \cdot l)$ under the assumption that $|\delta| \geq |S|$, where $l$ is the number of fuzzy values used in~$\mS$ plus~1. 
In this work, we provide efficient algorithms for computing the partition corresponding to the greatest crisp bisimulation of a finite \NFTS $\mS$, as well as the compact fuzzy partition corresponding to the greatest fuzzy bisimulation of $\mS$ under the G\"odel semantics. Their time complexities are of the order $O((\size(\delta) \log{l} + |S|) \log{(|S| + |\delta|)})$, where $l$ is the number of fuzzy values used in~$\mS$ plus~2. When $|\delta| \geq |S|$, this order is within $O(|S| \cdot |\delta| \cdot \log^2{|\delta|})$. The reduction of time complexity from $O(|S|^4 \cdot |\delta|^2)$ and $O(|S|^4 \cdot |\delta|^2 \cdot l)$ to $O(|S| \cdot |\delta| \cdot \log^2{|\delta|})$ is a significant contribution of this work. 
In addition, we introduce nondeterministic fuzzy labeled transition systems, which extend \NFTSs with fuzzy state labels, and we define and provide results on simulations and bisimulations between them.
\end{abstract}


\end{frontmatter}

%===============================================================================

\section{Introduction}
\label{section:intro}

Fuzzy transition systems (\FTSs) offer a robust framework for modeling and analyzing systems with inherent uncertainties and imprecision, which are prevalent in real-world scenarios. They extend traditional transition systems by incorporating fuzzy transitions, which enable more nuanced state changes. 
In~\cite{CaoCK11} Cao {\em et al.} introduced and studied (crisp) bisimulations between \FTSs. 
In~\cite{DBLP:journals/kbs/IgnjatovicCS13} Ignjatovi\'c {\em et al.} studied subsystems of \FTSs via fuzzy relation inequalities and equations. 
In~\cite{DBLP:journals/ijar/PanC0C14} Pan {\em et al.} introduced and studied fuzzy simulations for 
fuzzy labeled transition systems (\FLTSs), which extend \FTSs with fuzzy state labels. 
In~\cite{DBLP:journals/ijar/PanLC15} Pan {\em et al.} introduced and studied fuzzy/crisp simulations for quantitative transition systems, which are variants of \FLTSs. 
In~\cite{fss/WuD16} Wu {\em et al.} provided logical characterizations of (crisp) simulations and bisimulations for \FTSs. 

In~\cite{CaoSWC13} Cao {\em et al.} studied nondeterministic fuzzy transition systems (\NFTSs), which are a generalization of \FTSs, stating that ``nondeterminism is essential for modeling scheduling freedom, implementation freedom, the external environment, and incomplete information''. They introduced and studied the behavioral distance between states of a finite \NFTS, which measures the dissimilarity between the states. They also defined (crisp) bisimulations of an \NFTS and proved that two states are bisimilar (i.e., form a pair belonging to the greatest bisimulation) iff the behavioral distance between them is~0. 

%---------------------------------------------------------------------

Bisimulations are robust formal notions for examining the equivalence or similarity between states. 
Two important works on bisimulations for \NFTSs are \cite{DBLP:journals/fss/WuCBD18,DBLP:journals/tfs/QiaoZF23}. 
In~\cite{DBLP:journals/fss/WuCBD18} Wu {\em et al.} provided algorithmic and logical
 characterizations of (crisp) bisimulations for \NFTSs. 
They gave an algorithm for checking whether two states of a finite \NFTS $\mS = \tuple{S, \SA, \delta}$ are bisimilar. (Here, $S$, $A$ and $\delta$ are the set of states, the set of actions and the transition relation of $\mS$, respectively.) The algorithm runs in time of the order $O(|S|^4 \cdot |\delta|^2)$, under the assumption that $|\delta| \geq |S|$.\footnote{Proposition~4.3 of~\cite{DBLP:journals/fss/WuCBD18} and its proof should be made precise by adding the assumption that \mbox{$|\!\!\to\!\!| \geq |S|$}, which means $|\delta| \geq |S|$.} 
In~\cite{DBLP:journals/tfs/QiaoZF23} Qiao {\em et al.} introduced and studied fuzzy bisimulations for \NFTSs. They gave fixed-point and logical characterizations of such bisimulations. They also provided an algorithm for computing the greatest fuzzy bisimulation of a finite \NFTS $\mS$ when the used operator $\fand$ is the G\"odel or \L{}ukasiewicz t-norm. The complexity analysis given in \cite{DBLP:journals/tfs/QiaoZF23} states that, when $\fand$ is the G\"odel t-norm, the algorithm runs in time of the order \mbox{$O(|S|^6 \cdot |\!\!\to\!\!|^2 \cdot |A| \cdot l)$}, where \mbox{$|\!\!\to\!\!|$} is the maximum number of transitions outgoing from a state and $l$ is the number of fuzzy values used in~$\mS$ plus~1. A tighter analysis of the complexity of that algorithm would give $O(|S|^4 \cdot |\delta|^2 \cdot l)$, under the assumption that $|\delta| \geq |S|$. 

Other notable works on bisimulations for \NFTSs concern 
distribution-based behavioral distance for \NFTSs~\cite{tfs/WuD18}, 
group-by-group fuzzy\footnote{In contrast to the name, group-by-group fuzzy bisimulations defined in \cite{DBLP:journals/ijar/WuCHC18} are crisp relations.} bisimulations for \NFTSs~\cite{DBLP:journals/ijar/WuCHC18}, 
approximate bisimulations for \NFTSs~\cite{fss/QiaoZP23}, 
distribution-based limited fuzzy bisimulations for \NFTSs~\cite{jfi/QiaoF024}, 
as well as modeling and specification of nondeterministic fuzzy discrete-event systems~\cite{DBLP:books/sp/20/CaoECP20}. 

%---------------------------------------------------------------------

The main aim of this work is to develop efficient algorithms for computing the greatest crisp/fuzzy bisimulation of a finite \NFTS. 
We are motivated to design algorithms with a complexity order much lower than the ones of the algorithms provided in \cite{DBLP:journals/fss/WuCBD18,DBLP:journals/tfs/QiaoZF23}. Apart from these works, which have been discussed above, other closely related works are 
\cite{conf/tase/BuWC17, tcs/Chen0C18}. 
In~\cite{conf/tase/BuWC17} Bu {\em et al.} provided an algorithm with the time complexity order $O(|S|^5 \cdot |\delta|^3 \cdot \log{|\delta|})$ for computing the behavioral distance between states of a finite \NFTS. 
In~\cite{tcs/Chen0C18} Chen {\em et al.} provided polynomial time algorithms for computing the behavioral distance between states of a finite \NFTS (also for the case with discounting), without giving a concrete complexity order. 
As stated before, the behavioral distance between states is closely related to bisimulations for \NFTSs~\cite{CaoSWC13}.\footnote{We have the conjecture that the behavioral distance $d_f$~\cite{CaoSWC13} is the complement of a fuzzy relation between the crisp bisimilarity $Z_c$ and the fuzzy bisimilarity $Z_f$ w.r.t.\ the G\"odel semantics. That is, $Z_c(s,t) \leq 1 - d_f(s,t) \leq Z_f(s,t)$ for all states $s$ and $t$ of a given \NFTS $\mS$, where $Z_c$ (resp.\ $Z_f$) is the greatest crisp bisimulation (resp.\ fuzzy bisimulation w.r.t.\ the G\"odel semantics) of~$\mS$.} 

%---------------------------------------------------------------------

In this work, we provide efficient algorithms for computing the partition corresponding to the greatest crisp bisimulation of a finite \NFTS $\mS = \tuple{S, \SA, \delta}$, as well as the compact fuzzy partition corresponding to the greatest fuzzy bisimulation of $\mS$ when $\fand$ is the G\"odel t-norm. Their time complexities are of the order $O((\size(\delta) \log{l} + |S|) \log{(|S| + |\delta|)})$, where $\size(\delta)$ is the amount of data used to specify the transition relation $\delta$ and $l$ is the number of fuzzy values used in~$\mS$ plus~2. When $|\delta| \geq |S|$, this order is within $O(|S| \cdot |\delta| \cdot \log^2{|\delta|})$. 

%---------------------------------------------------------------------

The reduction of time complexity from $O(|S|^4 \cdot |\delta|^2)$~\cite{DBLP:journals/fss/WuCBD18} and $O(|S|^4 \cdot |\delta|^2 \cdot l)$~\cite{DBLP:journals/tfs/QiaoZF23} to $O(|S| \cdot |\delta| \cdot \log^2{|\delta|})$ is a significant contribution of this work. 
Regarding the case $|\delta| \geq |S|$ and taking $10^9$ as the limit for the number of steps an algorithm can execute using a laptop, the algorithms given in \cite{DBLP:journals/fss/WuCBD18,DBLP:journals/tfs/QiaoZF23} cannot deal with \NFTSs having 32 states or more, while our algorithms can deal with \NFTSs having about 2765 states.\footnote{We have $n^6 > 10^9$ for $n \geq 32$, and $n^2 \log^2{n} < 10^9$ for $n \leq 2765$. For simplicity, we ignore the constant factors hidden in the $O(\cdot)$ notation.} 
More realistically, since our algorithms have the time complexity of the order $O((\size(\delta) \log{l} + |S|) \log{(|S| + |\delta|)})$, they execute more than $10^9$ steps only when $\size(\delta)$ is really too big.

%---------------------------------------------------------------------

As a further contribution, we introduce nondeterministic fuzzy labeled transition systems (\NFLTSs), which extend \NFTSs with fuzzy state labels, and we define and provide results on simulations and bisimulations between them. In particular, our above mentioned algorithms are still correct when taking a finite \NFLTS as the input instead of a finite \NFTS. Furthermore, we present efficient algorithms for computing the greatest crisp (resp.\ fuzzy) simulation between two finite \NFLTSs $\mS$ and $\mS'$ (under the G\"odel semantics in the case of fuzzy simulation). 
Their time complexities are of the order $O((m+n)n)$, where $m = \size(\delta) + \size(\delta')$ and $n = |S| + |S'| + |\delta| + |\delta'|$, with $S$ and $\delta$ (resp.\ $S'$ and $\delta'$) being the set of states and the transition relation of $\mS$ (resp.\ $\mS'$). 

%---------------------------------------------------------------------

The rest of this work is structured as follows. In Section~\ref{section: prel}, we recall the definitions of fuzzy sets and relations, the compact fuzzy partition corresponding to a fuzzy equivalence relation~\cite{DBLP:journals/isci/Nguyen23}, the formal notions of crisp/fuzzy bisimulations for \NFTSs \cite{DBLP:journals/fss/WuCBD18,DBLP:journals/tfs/QiaoZF23}, the definition of fuzzy labeled graphs (\FLG{}s), and the notions of crisp/fuzzy bisimulations for \FLG{}s \cite{DBLP:journals/isci/Nguyen23,DBLP:journals/ijar/NguyenT24}. In Section~\ref{section: transformation}, we present a transformation of an \NFTS to an \FLG. 
By using that transformation, in Section~\ref{section: computation}, we present our algorithms for computing the greatest crisp/fuzzy bisimulation of a finite \NFTS. In Section~\ref{section: extension}, we present our results on \NFLTSs. Section~\ref{section: conc} contains conclusions.

%===============================================================================

\section{Preliminaries}
\label{section: prel}

By $\land$ and $\lor$ we denote the functions $\min$ and $\max$ on the unit interval $[0,1]$. 
For $\Gamma \subseteq [0,1]$, by $\bigwedge\! \Gamma$ and $\bigvee\! \Gamma$ we denote the infimum and supremum of $\Gamma$, respectively. 
If not stated otherwise, let $\fand$ denote any left-continuous t-norm and $\fto$ the corresponding residuum (see, e.g., \cite{Hajek1998,Belohlavek2002}). 
Let $\fequiv$ be the binary operator on $[0,1]$ defined by \mbox{$(x \fequiv y) = (x \fto y) \land (y \fto x)$}. The G\"odel t-norm~$\fand$ is the same as $\land$, which is continuous, and its corresponding residuum is defined by: $(x \fto y) = 1$ if $x \leq y$, and $(x \fto y) = y$ otherwise. 

Given a set $X$, a {\em fuzzy subset} of $X$ is any function from $X$ to $[0,1]$. It is also called a {\em fuzzy set}. By $\mF(X)$ we denote the set of all fuzzy subsets of~$X$. For $\mu \in \mF(X)$ and $U \subseteq X$, we denote $\support(\mu) = \{x \in X \mid \mu(x) > 0\}$ and $\mu(U) = \bigvee_{x \in U} \mu(x)$. 
Given $\mu, \nu \in \mF(X)$, we say that $\mu$ is {\em greater than or equal to} $\nu$, denoted by $\nu \leq \mu$, if $\nu(x) \leq \mu(x)$ for all $x \in X$. 

For $\{a_i\}_{i \in I} \subseteq [0,1]$, we write $\{x_i\!:\!a_i\}_{i \in I}$ or $\{x_1\!:\!a_1$, \ldots, $x_n\!:\!a_n\}$ when $I = 1..n$ to denote the fuzzy set $\mu$ specified by: $\support(\mu) \subseteq \{x_i\}_{i \in I}$ and $\mu(x_i) = a_i$ for $i \in I$. 

A fuzzy subset of $X \times Y$ is called a {\em fuzzy relation} between $X$ and $Y$. Given fuzzy relations $r \in \mF(X \times Y)$ and $s \in \mF(Y \times Z)$, the {\em converse} of $r$ is $r^{-1} \in \mF(Y \times X)$ specified by $r^{-1}(y,x) = r(x,y)$, for $x \in X$ and $y \in Y$, and the {\em composition} of~$r$ and~$s$ (w.r.t.~$\fand$) is $(r \circ s) \in \mF(X \times Z)$ specified by $(r \circ s)(x,z) = \bigvee_{y \in Y} r(x,y) \fand s(y,z)$, for $x \in X$ and $z \in Z$. 
A fuzzy relation $r \in \mF(X \times X)$ is called a fuzzy relation on~$X$. It is a {\em fuzzy equivalence relation} on $X$ (w.r.t.~$\fand$) if it is {\em reflexive} (i.e., $r(x,x) = 1$ for all $x \in X$), {\em symmetric} (i.e., $r = r^{-1}$) and {\em transitive} (i.e., $r \circ r \leq r$). 

\subsection{Compact fuzzy partitions}

The (traditional) fuzzy partition corresponding to a fuzzy equivalence relation $r$ on $X$ is usually defined to be the set $\{\mu \in \mF(X) \mid$ there exists $x \in X$ such that $\mu(y) = r(x,y)$ for all $y \in X\}$ \cite{OVCHINNIKOV1991107,DBLP:conf/ismvl/Schmechel95,DBLP:journals/isci/BaetsCK98,DBLP:journals/fss/CiricIB07}. In~\cite{DBLP:journals/isci/Nguyen23} we introduced a new notion of the fuzzy partition that corresponds to a fuzzy equivalence relation on a finite set for the case where $\fand$ is the G\"odel t-norm. We recall it below, extending its name with the
 word ``compact''. 

\begin{definition}\label{def: HFJAA}
Consider the case where $\fand$ is the G\"odel t-norm. 
Given a finite set $X$ and a fuzzy equivalence relation $r \in \mF(X \times X)$, the {\em compact fuzzy partition corresponding to $r$} is the data structure $B$ defined inductively as follows:
\begin{itemize}
\item if $r(x,x') = 1$ for all $x,x' \in X$, then $B$ has two attributes, $B.\degree = 1$ and $B.\elements = X$, $B$ is also called a {\em crisp block} and denoted by $X_1$;
\item else:
    \begin{itemize}
	\item let $d = \bigwedge_{x,x' \in X} r(x,x')$;
	\item let $\sim$ be the equivalence relation on $X$ such that $x \sim x'$ iff $r(x,x') > d$;
	\item let $\{Y_1,\ldots,Y_n\}$ be the (crisp) partition of $X$ corresponding to~$\sim$;
	\item let $r_i$ be the restriction of $r$ to $Y_i \times Y_i$ and $B_i$ the compact fuzzy partition corresponding to $r_i$, for $1 \leq i \leq n$;
	\item $B$ has two attributes, $B.\degree = d$ and $B.\subblocks = \{B_1,\ldots,B_n\}$, $B$ is also called a {\em fuzzy block} and denoted by $\{B_1,\ldots,B_n\}_d$.
\myend
	\end{itemize}
\end{itemize}
\end{definition}

\begin{example}
Let $X = \{x_1, x_2, \ldots, x_7\}$ and let $r: X \times X \to [0,1]$ be the fuzzy relation specified by the following table.
\[
\begin{array}{|c||c|c|c|c|c|c|c|}
\hline
r & x_1 & x_2 & x_3 & x_4 & x_5 & x_6 & x_7 \\
\hline\hline
x_1 & 1 & 0.4 & 0.4 & 0.4 & 0.1 & 0.1 & 0 \\
\hline
x_2 & 0.4 & 1 & 0.6 & 0.6 & 0.1 & 0.1 & 0 \\
\hline
x_3 & 0.4 & 0.6 & 1 & 1 & 0.1 & 0.1 & 0 \\
\hline
x_4 & 0.4 & 0.6 & 1 & 1 & 0.1 & 0.1 & 0 \\
\hline
x_5 & 0.1 & 0.1 & 0.1 & 0.1 & 1 & 0.3 & 0 \\
\hline
x_6 & 0.1 & 0.1 & 0.1 & 0.1 & 0.3 & 1 & 0 \\
\hline
x_7 & 0 & 0 & 0 & 0 & 0 & 0 & 1 \\
\hline
\end{array}
\]

It is a fuzzy equivalence relation on $X$ w.r.t.\ the G\"odel semantics. 
The traditional fuzzy partition of $X$ that corresponds to $r$ is the set $\{\mu_1,\mu_2,\mu_{3,4},\mu_5,\mu_6,\mu_7\} \subset \mF(X)$ specified by: $\mu_i(x) = r(x_i, x)$ for $i \in \{1,2,5,6,7\}$ and $\mu_{3,4}(x) = r(x_3,x) = r(x_4,x)$, for $x \in X$.
The compact fuzzy partition corresponding to $r$ is the data structure denoted by 
\[
\{\{\{\{x_1\}_1, \{\{x_2\}_1,\{x_3,x_4\}_1\}_{0.6}\}_{0.4}, \{\{x_5\}_1, \{x_6\}_1\}_{0.3}\}_{0.1}, \{x_7\}_1\}_0.
\]
The advantage of this kind of data structure is that it uses only linear space. 
\myend
\end{example}

\subsection{Nondeterministic fuzzy transition systems}

A {\em nondeterministic fuzzy transition system} (\NFTS) is a structure $\mS = \tuple{S, \SA, \delta}$, where $S$ is a non-empty set of states, $\SA$ a non-empty set of actions, and $\delta \subseteq S \times \SA \times \mF(S)$ a~set called the transition relation. 
It is {\em finite} if all the components $S$, $\SA$ and $\delta$ are finite. 
We denote 
\[ \deltatt = \{\mu \mid \tuple{s,a,\mu} \in \delta \textrm{ for some $s$ and $a$}\} \]
and define the size of $\delta$ as follows, where $|X|$ denotes the cardinality of~$X$: 
\[ \size(\delta) = |\delta| + \sum_{\mu \in \deltatt} |\support(\mu)|. \]

\begin{Definition}[\cite{DBLP:journals/fss/WuCBD18}]\label{def: JHFHJ}
Given $R \subseteq S \times S$, the {\em lifted relation} of $R$ is the subset $R^\dag$ of $\mF(S) \times \mF(S)$ such that $\mu R^\dag \mu'$ iff there exists a function $e : S \times S \to [0, 1]$ that satisfies the following conditions:
\begin{itemize}
\item $\mu(s) = \bigvee_{s' \in S} e(s, s')$, for every $s \in S$;
\item $\mu'(s') = \bigvee_{s \in S} e(s, s')$, for every $s' \in S$;
\item $e(s, s') = 0$ if $\tuple{s, s'} \notin R$.
\myend
\end{itemize}
\end{Definition}

Given $R \subseteq S \times S$ and $s,s' \in S$, $\overrightarrow{R_s}$ denotes the set $\{s' \in S \mid$ $s R s'\}$, whereas $\overleftarrow{R_{s'}}$ denotes the set $\{s \in S \mid s R s'\}$. 
It is proved in~\cite[Theorem~3.2]{DBLP:journals/fss/WuCBD18} that $\mu R^\dag \mu'$ iff, for every $s,s' \in S$, 
\begin{equation}\label{eq: lifted relation}
\mu(s) \leq \mu'(\overrightarrow{R_s})\ \ \textrm{and}\ \ \mu'(s') \leq \mu(\overleftarrow{R_{s'}}), 
\end{equation}
and as the above mentioned function $e : S \times S \to [0, 1]$ we can take 
\[ \lambda\tuple{s,s'}.(\textrm{if $s R s'$ then $\min(\mu(s),\mu'(s'))$ else 0}). \]
In addition, $(R^{-1})^\dag = (R^\dag)^{-1}$ \cite[Lemma~3.3]{DBLP:journals/fss/WuCBD18}. 

The following notion of bisimulation comes from~\cite[Definition~3.5]{DBLP:journals/fss/WuCBD18}. 

\begin{definition}\label{def: KJFHA}
Let $\mS = \tuple{S, \SA, \delta}$ be an \NFTS. A relation $R \subseteq S \times S$ is called a {\em crisp auto-bisimulation of $\mS$} (or a {\em crisp bisimulation of $\mS$} for short) if, for every $\tuple{s,s'} \in R$, 
\begin{enumerate}
\item[(a)] for every $\tuple{s,a,\mu} \in \delta$, there exists $\tuple{s',a,\mu'} \in \delta$ such that $\mu R^\dag \mu'$;
\item[(b)] for every $\tuple{s',a,\mu'} \in \delta$, there exists $\tuple{s,a,\mu} \in \delta$ such that $\mu R^\dag \mu'$.
\myend
\end{enumerate}
\end{definition}

Wu {\em et al.}~\cite{DBLP:journals/fss/WuCBD18} proved that the greatest crisp bisimulation of any \NFTS exists and is an equivalence relation. 

The following notion of the lifted relation of a fuzzy relation $R \in \mF(S \times S)$ comes from~\cite[Definition~4]{DBLP:journals/tfs/QiaoZF23}. We use the notation $R^\ddag$ to make it different from the lifted relation $R^\dag$ of a crisp relation $R \subseteq S \times S$ (Definition~\ref{def: JHFHJ}). 

\begin{definition}
Given a fuzzy relation $R$ on $S$, the {\em lifted relation} $R^\ddag$ (w.r.t.~$\fand$) is the fuzzy relation on $\mF(S)$ defined as follows, for any $\mu, \mu' \in \mF(S)$: 
\begin{equation}
R^\ddag(\mu, \mu') = \big[\!\bigwedge_{s \in S} (\mu(s) \fto \bigvee_{s' \in S} (R(s,s') \fand \mu'(s')))\big] \land \big[\!\bigwedge_{s' \in S} (\mu'(s') \fto \bigvee_{s \in S} (R(s,s') \fand \mu(s)))\big]. \label{eq: JHFKJ}
\end{equation} 
\myend
\end{definition}

It is proved in \cite[Lemma~2]{DBLP:journals/tfs/QiaoZF23} that $(R^{-1})^\ddag = (R^\ddag)^{-1}$. 

The following notion of a fuzzy bisimulation is a corrected version of the one from~\cite[Definition~5]{DBLP:journals/tfs/QiaoZF23}. 

\begin{definition}\label{def: HGLAK}
Let $\mS = \tuple{S, \SA, \delta}$ be an \NFTS. A fuzzy relation $R \in \mF(S \times S)$ is called a {\em fuzzy auto-bisimulation of $\mS$ w.r.t.~$\fand$} (or a {\em fuzzy bisimulation of $\mS$} for short) if, for every $s,s' \in S$ with $R(s,s') > 0$,\footnote{Definition~5 of~\cite{DBLP:journals/tfs/QiaoZF23} does not require the condition $R(s,s') > 0$. This is a mistake, because without this condition an \NFTS may not have any fuzzy bisimulation, which contradicts the other results of~\cite{DBLP:journals/tfs/QiaoZF23}.} 
\begin{enumerate}
\item[(a)] for every $\tuple{s,a,\mu} \in \delta$, there exists $\tuple{s',a,\mu'} \in \delta$ such that $R(s,s') \leq R^\ddag(\mu,\mu')$;
\item[(b)] for every $\tuple{s',a,\mu'} \in \delta$, there exists $\tuple{s,a,\mu} \in \delta$ such that $R(s,s') \leq R^\ddag(\mu,\mu')$.
\myend
\end{enumerate}
\end{definition}

The greatest fuzzy bisimulation of $\mS$ always exists and is called the {\em fuzzy bisimilarity} of $\mS$.\footnote{See \cite[Proposition~5]{DBLP:journals/tfs/QiaoZF23} and take into account the above mentioned correction.} 
It is a fuzzy equivalence relation~\cite[Proposition~3]{DBLP:journals/tfs/QiaoZF23}. 

%===============================================================================

\subsection{Fuzzy labeled graphs}

A {\em fuzzy labeled graph} (\FLG) \cite{DBLP:journals/isci/Nguyen23,DBLP:journals/ijar/NguyenT24} is a structure $G = \tuple{V, E, L, \SV, \SE}$, with a set $V$ of vertices, a set $\SV$ of vertex labels, a set $\SE$ of edge labels, a fuzzy set \mbox{$E \in \mF(V \times \SE \times V)$} of labeled edges, and a function $L: V \to \mF(\SV)$ that labels vertices. 
If $V$, $\SV$ and $\SE$ are finite, then $G$ is {\em finite}. 
%
%For $x,y \in V$ and $r \in \SE$, $L(x)$ is the fuzzy set of labels of~$x$, whereas a positive value $E(x,r,y)$ represents the degree of an edge labeled by~$r$ from $x$ to~$y$. 

\begin{Definition}[\cite{DBLP:journals/ijar/NguyenT24}]\label{def: DSHGQ}
A {\em crisp auto-bisimulation} (or {\em crisp bisimulation} for short) of an \FLG $G = \tuple{V, E, L, \SV, \SE}$ is a non-empty relation $Z \subseteq V \times V$ such that, for every $\tuple{x,x'} \in Z$ and $r \in \SE$, 
\begin{enumerate}
\item[(a)] $L(x) = L(x')$, 
\item[(b)] for every $y \in V$ with \mbox{$E(x,r,y) > 0$}, there exists $y'\in V$ such that $y Z y'$ and \mbox{$E(x,r,y) \leq E(x',r,y')$}, 
\item[(c)] for every $y' \in V$ with \mbox{$E(x',r,y') > 0$}, there exists $y \in V$ such that $y Z y'$ and \mbox{$E(x',r,y') \leq E(x,r,y)$}.
\myend
\end{enumerate}
\end{Definition}

The greatest crisp bisimulation of an \FLG always exists and is an equivalence relation~\cite[Corollary~2.2]{DBLP:journals/ijar/NguyenT24}. Nguyen and Tran \cite{DBLP:journals/ijar/NguyenT24} provided an efficient algorithm for computing the partition corresponding to the greatest crisp bisimulation of a finite \FLG $G = \tuple{V, E, L, \SV, \SE}$, with the complexity order $O((m \log{l} + n) \log{n})$, where $n = |V|$, $m = |\support(E)|$ and $l = |\{E(e) : e \in \support(E)\} \cup \{0,1\}|$. When $m \geq n$, this complexity order is the same as $O(m\cdot\log{n}\cdot\log{l})$, which is within $O(m\log^2{n})$.

\begin{Definition}[\cite{DBLP:journals/isci/Nguyen23}]\label{def: HGAKF}
A {\em fuzzy auto-bisimulation} (or {\em fuzzy bisimulation} for short) of an \FLG $G = \tuple{V, E, L, \SV, \SE}$ (w.r.t.~$\fand$) is a fuzzy relation $Z \in \mF(V \times V)$ satisfying the following conditions, for every $p \in \SV$, $r \in \SE$ and every possible values for the free variables:
\begin{eqnarray}
\!\!\!\!\!\!\!\!\!\!\!\!\!&&
 Z(x,x') \leq (L(x)(p) \fequiv L(x')(p)) \label{eq: FB1} \\
\!\!\!\!\!\!\!\!\!\!\!\!\!&& \E y'\! \in\!V (Z(x,x') \!\fand\! E(x,r,y) \leq E(x',r,y') \!\fand\! Z(y,y')) \label{eq: FB2} \\
\!\!\!\!\!\!\!\!\!\!\!\!\!&& \E y \in\!V (Z(x,x') \!\fand\! E(x',r,y') \leq E(x,r,y) \!\fand\! Z(y,y')). \label{eq: FB3}
\end{eqnarray}
\myend
\end{Definition}

It is known that, if $\fand$ is continuous, then the greatest fuzzy bisimulation of any finite \FLG exists and is a fuzzy equivalence relation \cite[Corollary~5.3]{FBSML}. In \cite{DBLP:journals/isci/Nguyen23}, we provided an efficient algorithm for computing the compact fuzzy partition corresponding to the greatest fuzzy bisimulation of a finite \FLG $G$ in the case where $\fand$ is the G\"odel t-norm. Its complexity is of order $O((m\log{l} + n)\log{n})$, where $l$, $m$ and $n$ are as specified above. That algorithm directly yields another algorithm with the same complexity order $O((m\log{l} + n)\log{n})$ for computing the bisimilarity degree between two given vertices $x$ and $x'$ of $G$ (i.e., $Z(x,x')$ with $Z$ being the greatest fuzzy bisimulation of $G$ w.r.t.\ the G\"odel semantics). It also yields an algorithm with the complexity order $O(m\!\cdot\!\log{n}\!\cdot\!\log{l} + n^2)$ for explicitly computing the greatest fuzzy bisimulation of $G$ w.r.t.\ the G\"odel semantics. 

%===============================================================================

\section{Transforming nondeterministic fuzzy transition systems to fuzzy labeled graphs}
\label{section: transformation}

In this section, we define the notion of the \FLG corresponding to an \NFTS, formulate and prove the relationship between the greatest crisp (resp.\ fuzzy) bisimulation of an \NFTS and the greatest crisp (resp.\ fuzzy) bisimulation of its corresponding \FLG.

\begin{definition}\label{def: JHFSO 1}
Given an \NFTS $\mS = \tuple{S, \SA, \delta}$, the {\em \FLG corresponding to $\mS$} is the \FLG $G = \tuple{V, E, L, \SV, \SE}$ specified as follows:
\begin{itemize}
\item $\SV = \{s\}$ and $\SE = \SA \cup \{\varepsilon\}$, where $s$ stands for ``being a state'' and $\varepsilon \notin A$ stands for ``the empty action''; 
\item $V = S \cup \deltatt$ and $L: V \to \mF(\SV)$ is specified by $L(x)(s) = 1$ for $x \in S$, and $L(x)(s) = 0$ for $x \in \deltatt$; 
\item $E: V \times \SE \times V \to [0,1]$ is defined as follows:
    \begin{itemize}
    \item $E(s, a, \mu) = 1$, for $\tuple{s,a,\mu} \in \delta$, 
    \item $E(\mu, \varepsilon, t) = \mu(t)$, for $\mu \in \deltatt$ and $t \in S$,
    \item $E(x, r, y) = 0$ for the other triples $\tuple{x,r,y}$ (i.e., for $\tuple{x,r,y} \in V \times \SE \times V$ that neither belongs to $\delta$ nor is of the form $\tuple{\mu, \varepsilon, t}$ with $\mu \in \deltatt$ and $t \in S$). 
\myend
    \end{itemize}
\end{itemize}
\end{definition}

\begin{proposition}
Let $\mS = \tuple{S, \SA, \delta}$ be a finite \NFTS and $G = \tuple{V, E, L, \SV, \SE}$ the \FLG corresponding to $\mS$. Then,
$|V| = |S| + |\deltatt|$, 
$|\support(E)| = \size(\delta)$ and $\{E(e) : e \in \support(E)\}$ is the set of fuzzy values used in~$\mS$ extended with 1.
\end{proposition}

This proposition directly follows from Definition~\ref{def: JHFSO 1}. 

\begin{remark}\label{remark: JHFJS}
The cost of constructing the \FLG $G$ that corresponds to a given finite \NFTS $\mS = \tuple{S, \SA, \delta}$ depends on their data representation. Under typical assumptions (e.g., a computer word can be used to identify any state or action) and by using an appropriate data representation for $G$ and $\mS$ (e.g., a fuzzy set is stored by restricting to its support, fuzzy subsets of $S$ are identified by references and represented without duplicates), the cost is of the order \InitComplexity.
\end{remark}

\begin{example}\label{example: HGFJS}
Consider the following \NFTS $\mS = \tuple{S, \SA, \delta}$, 

\medskip

\newcommand{\distA}{2.4cm}
\newcommand{\distB}{1.5cm}
\begin{center}
\begin{tikzpicture}[->,>=stealth]
\node[circle,draw] (s1) {$s_1$};
\node (m1) [node distance=\distA, right of=s1] {$\mu_1$};
\node[circle,draw] (s2) [node distance=\distA, right of=m1] {$s_2$};
\node (m3) [node distance=\distA, right of=s2] {$\mu_3$};
\node (m2) [node distance=\distB, below of=s1] {$\mu_2$};
\node[circle,draw] (s3) [node distance=\distA, right of=m2] {$s_3$};
\node[circle,draw] (s4) [node distance=\distA, right of=s3] {$s_4$};
\node[circle,draw] (s5) [node distance=\distA, right of=s4] {$s_5$};
\draw (s1) to node [above]{$a$} (m1);
\draw (m1) to node [above]{\footnotesize{0.5}} (s2);
\draw (s2) to node [above]{$a$} (m3);
\draw (s1) to node [left]{$a$} (m2);
\path (m1) edge [bend right=20] node [left]{\footnotesize{0.8}} (s3);
\path (s3) edge [bend right=20] node [right]{$b$} (m1);
\draw (m2) to node [above]{\footnotesize{0.6}} (s3);
\draw (m3) to node [left, xshift=1mm, yshift=1.5mm]{\footnotesize{0.7}} (s4);
\draw (s4) to node [right, xshift=0mm, yshift=1.5mm]{$b$} (m1);
\path (m2) edge [bend right=20] node [below]{\footnotesize{0.4}} (s5);
\path (m3) edge [bend right=20] node [left, xshift=0.0mm, yshift=-1mm]{\footnotesize{0.9}} (s5);
\path (s5) edge [bend right=20] node [right, yshift=-1mm]{$a$} (m3);
\end{tikzpicture}
\end{center}	

\medskip

\noindent
which is specified by: $S = \{s_1, s_2, s_3, s_4, s_5\}$, $A = \{a,b\}$, $\delta = \{\tuple{s_1,a,\mu_1}$, $\tuple{s_1,a,\mu_2}$, $\tuple{s_2,a,\mu_3}$, $\tuple{s_3,b,\mu_1}$, $\tuple{s_4,b,\mu_1}$, $\tuple{s_5,a,\mu_3}\}$, with $\mu_1 = \{s_2\!:\!0.5$, $s_3\!:\!0.8\}$, $\mu_2 = \{s_3\!:\!0.6$, $s_5\!:\!0.4\}$ and $\mu_3 = \{s_4\!:\!0.7$, $s_5\!:\!0.9\}$. 
The \FLG $G = \tuple{V, E, L, \SV, \SE}$ that corresponds to $\mS$ is illustrated below 

\medskip

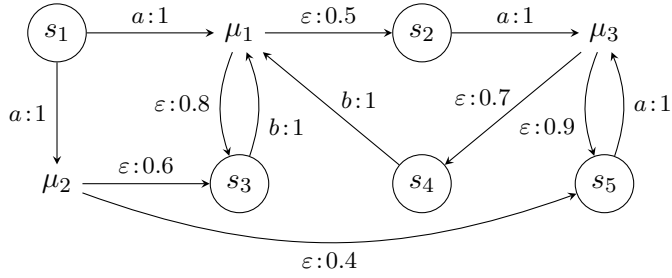
\renewcommand{\distA}{2.4cm}
\renewcommand{\distB}{2.0cm}
\begin{center}
\begin{tikzpicture}[->,>=stealth]
\node[circle,draw] (s1) {$s_1$};
\node (m1) [node distance=\distA, right of=s1] {$\mu_1$};
\node[circle,draw] (s2) [node distance=\distA, right of=m1] {$s_2$};
\node (m3) [node distance=\distA, right of=s2] {$\mu_3$};
\node (m2) [node distance=\distB, below of=s1] {$\mu_2$};
\node[circle,draw] (s3) [node distance=\distA, right of=m2] {$s_3$};
\node[circle,draw] (s4) [node distance=\distA, right of=s3] {$s_4$};
\node[circle,draw] (s5) [node distance=\distA, right of=s4] {$s_5$};
\draw (s1) to node [above]{\footnotesize{$a\!:\!1$}} (m1);
\draw (m1) to node [above]{\footnotesize{$\varepsilon\!:\!0.5$}} (s2);
\draw (s2) to node [above]{\footnotesize{$a\!:\!1$}} (m3);
\draw (s1) to node [left]{\footnotesize{$a\!:\!1$}} (m2);
\path (m1) edge [bend right=20] node [left]{\footnotesize{$\varepsilon\!:\!0.8$}} (s3);
\path (s3) edge [bend right=20] node [right, yshift=-3mm]{\footnotesize{$b\!:\!1$}} (m1);
\draw (m2) to node [above]{\footnotesize{$\varepsilon\!:\!0.6$}} (s3);
\draw (m3) to node [left, xshift=1mm, yshift=1.5mm]{\footnotesize{$\varepsilon\!:\!0.7$}} (s4);
\draw (s4) to node [right, yshift=1.2mm]{\footnotesize{$b\!:\!1$}} (m1);
\path (m2) edge [bend right=20] node [below]{\footnotesize{$\varepsilon\!:\!0.4$}} (s5);
\path (m3) edge [bend right=20] node [left, yshift=-3mm]{\footnotesize{$\varepsilon\!:\!0.9$}} (s5);
\path (s5) edge [bend right=20] node [right]{\footnotesize{$a\!:\!1$}} (m3);
\end{tikzpicture}
\end{center}

\medskip

\noindent
and has $V = S \cup \deltatt = S \cup \{\mu_1, \mu_2, \mu_3\}$. Examples of edges of $G$ are: $E(s_1,a,\mu_1) = 1$, $E(\mu_1,\varepsilon,s_3) = 0.8$.  
\myend
\end{example}

\begin{theorem}\label{theorem: HFLWA}
Let $\mS = \tuple{S, \SA, \delta}$ be a finite \NFTS and $G = \tuple{V, E, L, \SV, \SE}$ the \FLG corresponding to $\mS$. 
\begin{enumerate}
\item If $R$ is a crisp bisimulation of $\mS$, then the following relation $Z$ is a crisp bisimulation of $G$:
\begin{equation}\label{eq: JHDJH}
Z = R \cup \{\tuple{\mu,\mu'} \in \deltatt\!\times\deltatt \mid \mu R^\dag \mu'\}.
\end{equation}
\item If $Z$ is a crisp bisimulation of $G$, then $R = Z \cap (S \times S)$ is a crisp bisimulation of $\mS$.  
\end{enumerate}
\end{theorem}

\begin{proof}
Consider the first assertion and assume that $R$ is a crisp bisimulation of $\mS$ and $Z$ is defined by~\eqref{eq: JHDJH}. We need to show that $Z$ is a crisp bisimulation of $G$. Let $\tuple{x,x'} \in Z$ and $r \in \SE$. 

If $x R x'$, then $x,x' \in S$, otherwise $x,x' \in \deltatt$. In both of the cases, according to the definition of~$L$, we have $L(x) = L(x')$. That is, the condition~(a) of Definition~\ref{def: DSHGQ} holds. 

Consider the condition~(b) of Definition~\ref{def: DSHGQ} and let $y \in V$ with \mbox{$E(x,r,y) > 0$}. 
Consider the case $x \in S$. We must have $r \in A$, $\tuple{x,r,y} \in \delta$ and $E(x,r,y) = 1$. 
Since $x Z x'$, we also have $x R x'$. Since $R$ is a crisp bisimulation of $\mS$, it follows that there exists $\tuple{x',r,y'} \in \delta$ such that $y R^\dag y'$. Thus, $y Z y'$ and $E(x', r, y') = 1 = E(x,r,y)$. 
Now consider the case $x \notin S$. We have $x \in \deltatt$. Since \mbox{$E(x,r,y) > 0$}, it follows that $r = \varepsilon$ and \mbox{$y \in S$}. Since $x Z x'$, we have $x' \in \deltatt$ and $x R^\dag x'$. Since $\mS$ is finite, by~\eqref{eq: lifted relation} with $\mu$, $\mu'$ and $s$ replaced by $x$, $x'$ and $y$, respectively, there exists $y' \in S$ such that $y R y'$ and $x(y) \leq x'(y')$. This implies that $y Z y'$ and \mbox{$E(x,r,y) \leq E(x',r,y')$}. Therefore, in both of the cases, the condition~(b) of Definition~\ref{def: DSHGQ} holds. 

Consider the condition~(c) of Definition~\ref{def: DSHGQ} and let $y' \in V$ with \mbox{$E(x',r,y') > 0$}. 
Consider the case $x' \in S$. We must have $r \in A$, $\tuple{x',r,y'} \in \delta$ and $E(x',r,y') = 1$. 
Since $x Z x'$, we also have $x R x'$. Since $R$ is a crisp bisimulation of $\mS$, it follows that there exists $\tuple{x,r,y} \in \delta$ such that $y R^\dag y'$. Thus, $y Z y'$ and $E(x, r, y)
 = 1 = E(x',r,y')$. 
Now consider the case $x' \notin S$. We have $x' \in \deltatt$. Since \mbox{$E(x',r,y') > 0$}, it follows that $r = \varepsilon$ and $y' \in S$. Since $x Z x'$, we have $x \in \deltatt$ and $x R^\dag x'$. Since $\mS$ is finite, by~\eqref{eq: lifted relation} with $\mu$, $\mu'$ and $s'$ replaced by $x$, $x'$ and $y'$, respectively, there exists $y \in S$ such that $y R y'$ and $x'(y') \leq x(y)$. This implies that $y Z y'$ and \mbox{$E(x',r,y') \leq E(x,r,y)$}. Therefore, in both of the cases, the condition~(c) of Definition~\ref{def: DSHGQ} holds.\footnote{In comparison with the previous paragraph, this one shows how a similar proof for the ``converse'' can be made detailed.} 

We have proved that $Z$ satisfies the conditions stated in Definition~\ref{def: DSHGQ} and is therefore a crisp bisimulation of~$G$.

Now consider the second assertion of the theorem and assume that $Z$ is a crisp bisimulation of~$G$. 
We need to prove that $R = Z \cap (S \times S)$ is a crisp bisimulation of~$\mS$. 
Let $\tuple{s,s'} \in R$ and $a \in A$. Thus, $s, s' \in S$ and $s Z s'$. 

Consider the condition~(a) of Definition~\ref{def: KJFHA} and let \mbox{$\tuple{s,a,\mu} \in \delta$}. Since $Z$ is a crisp bisimulation of~$G$, $s Z s'$ and $E(s,a,\mu) = 1$, there must exist $\mu' \in V$ such that $\mu Z \mu'$ and $E(s',a,\mu') = 1$. Hence, $\tuple{s',a,\mu'} \in \delta$. We prove that $\mu R^\dag \mu'$. By~\cite[Theorem~3.2]{DBLP:journals/fss/WuCBD18}, it suffices to prove that
\begin{eqnarray}
\!\!\!\!\!\!\!\!&& \textrm{for every $u \in S$ with $\mu(u) > 0$, there exists $u' \in S$ such that $u R u'$ and $\mu(u) \leq \mu'(u')$;} \label{eq: HJDJH1} \\[1ex]
\!\!\!\!\!\!\!\!&& \textrm{for every $u' \in S$ with $\mu'(u') > 0$, there exists $u \in S$ such that $u R u'$ and $\mu'(u') \leq \mu(u)$.} \label{eq: HJDJH2} 
\end{eqnarray}

Consider~\eqref{eq: HJDJH1} and let $u \in S$ with $\mu(u) > 0$. Since $Z$ is a crisp bisimulation of~$G$, $\mu Z \mu'$ and $E(\mu,\varepsilon,u) = \mu(u) > 0$, there must exist $u' \in V$ such that $u Z u'$ and $E(\mu,\varepsilon,u) \leq E(\mu',\varepsilon,u')$. This implies that $u' \in S$, $u R u'$ and $\mu(u) \leq \mu'(u')$. Therefore, \eqref{eq: HJDJH1} holds. 

The assertion~\eqref{eq: HJDJH2} can be proved analogously. Thus, we have proved that $R$ satisfies the condition~(a) of Definition~\ref{def: KJFHA}. Similarly, it can be proved that $R$ also satisfies the condition~(b) of Definition~\ref{def: KJFHA}. This completes the proof.
\myend
\end{proof}

\begin{corollary}\label{cor: KJDKW}
Let $\mS = \tuple{S, \SA, \delta}$ be a finite \NFTS and $G = \tuple{V, E, L, \SV, \SE}$ the \FLG corresponding to $\mS$. If $Z$ is the greatest crisp bisimulation of $G$, then $R = Z \cap (S \times S)$ is the greatest crisp bisimulation of $\mS$.
\end{corollary}

\begin{proof}
Let $Z$ be the greatest crisp bisimulation of $G$ and let $R = Z \cap (S \times S)$. By Theorem~\ref{theorem: HFLWA}, $R$ is a crisp bisimulation of $\mS$. Let $R'$ be an arbitrary crisp bisimulation of $\mS$ and let 
\[ Z' = R' \cup \{\tuple{\mu,\mu'} \in \deltatt\! \times \deltatt \mid \mu (R')^\dag \mu'\}. \]
By Theorem~\ref{theorem: HFLWA}, $Z'$ is a crisp bisimulation of $G$. Hence, $Z' \subseteq Z$ and 
\[ R' = Z' \cap (S \times S) \subseteq Z \cap (S \times S) = R. \] 
Therefore,  $R$ is the greatest crisp bisimulation of $\mS$.
\myend
\end{proof}

\begin{theorem}\label{theorem: JHFKJ}
Let $\mS = \tuple{S, \SA, \delta}$ be a finite \NFTS and $G = \tuple{V, E, L, \SV, \SE}$ the \FLG corresponding to $\mS$. 
\begin{enumerate}
\item If $R$ is a fuzzy bisimulation of $\mS$, then the fuzzy relation $Z$ on $V$ specified as follows is a fuzzy bisimulation of $G$:
    \begin{itemize}
    \item $Z(s,s') = R(s,s')$ for $s,s' \in S$, 
    \item \mbox{$Z(\mu,\mu') = R^\ddag(\mu,\mu')$} for $\mu, \mu' \in \deltatt$, 
    \item $Z(x,x') = 0$ for $\tuple{x,x'}$ from $(V \times V) - (S \times S) - (\deltatt \times \deltatt)$.
    \end{itemize}
\item If $Z$ is a fuzzy bisimulation of $G$, then $R = Z | _{S \times S}$ is a fuzzy bisimulation of $\mS$.  
\end{enumerate}
\end{theorem}

\begin{proof}
Consider the first assertion and assume that $R$ is a fuzzy bisimulation of $\mS$ and $Z$ is defined as in that assertion. We need to show that $Z$ is a fuzzy bisimulation of $G$. Let $x,x' \in V$ and $r \in \SE$. 

Consider the condition~\eqref{eq: FB1} with $p$ replaced by the unique element $s$ of $\SV$. If $x,x' \in S$, then $L(x)(s) = L(x')(s) = 1$ and the condition~\eqref{eq: FB1} clearly holds. If $x,x' \in \deltatt$, then $L(x)(s) = L(x')(s) = 0$ and the condition~\eqref{eq: FB1} also holds. For the other cases, we have $Z(x,x') = 0$ and the condition~\eqref{eq: FB1} also holds. 

Let $y \in V$ and consider the condition~\eqref{eq: FB2}. If $Z(x,x') = 0$ or $E(x,r,y) = 0$, then that condition clearly holds. So, we assume that $Z(x,x') > 0$ and $E(x,r,y) > 0$. Thus, $x,x' \in S$ or $x,x' \in \deltatt$. 
Consider the case $x,x' \in S$. We have $Z(x,x') = R(x,x')$. Since $x \in S$ and $E(x,r,y) > 0$, we must have $r \in A$, $\tuple{x,r,y} \in \delta$ and $E(x,r,y) = 1$. Since $R$ is a fuzzy bisimulation of $\mS$, there exists $\tuple{x',r,y'} \in \delta$ such that $R(x,x') \leq R^\ddag(y, y')$. Thus, $E(x',r,y') = 1$. Therefore, 
$Z(x,x') \fand E(x,r,y) = R(x,x') \leq R^\ddag(y, y') = E(x',r,y') \fand Z(y,y')$.
Now consider the case $x,x' \in \deltatt$. Since $E(x,r,y) > 0$ and $Z(x,x') > 0$, we have $r = \varepsilon$, $y \in S$, $x' \in \deltatt$ and $Z(x,x') = R^\ddag(x,x')$. Thus, $E(x,r,y) = x(y)$. Since $\mS$ is finite, by~\eqref{eq: JHFKJ} with $\mu$ and $\mu'$ replaced by $x$ and $x'$, respectively, there exists $y' \in S$ such that 
\[ R^\ddag(x,x') \leq (x(y) \fto (R(y,y') \fand x'(y'))), \]
which is equivalent to 
\[ R^\ddag(x,x') \fand x(y) \leq R(y,y') \fand x'(y'). \]
Therefore, 
\[ Z(x,x') \fand E(x,r,y) = R^\ddag(x,x') \fand x(y) \leq x'(y') \fand R(y,y') = E(x',r,y') \fand Z(y,y'). \] 

We have proved that, for any $y \in V$, the condition~\eqref{eq: FB2} holds. 
Similarly, for any $y' \in V$, it can be proved that the condition~\eqref{eq: FB3} holds. 
Therefore, $Z$ is a fuzzy bisimulation of~$G$.

Now consider the second assertion of the theorem and assume that $Z$ is a fuzzy bisimulation of~$G$. 
We need to prove that $R = Z|_{S \times S}$ is a fuzzy bisimulation of~$\mS$. 
Let $a \in A$ and $s,s' \in S$ with $R(s,s') > 0$. We have $Z(s,s') = R(s,s') > 0$. 

Consider the condition~(a) of Definition~\eqref{def: HGLAK} and let $\tuple{s,a,\mu} \in \delta$. 
Since $Z$ is a fuzzy bisimulation of $G$, there exists $\mu' \in V$ such that 
\begin{equation}\label{eq: JHFJH}
    Z(s,s') \fand E(s,a,\mu) \leq E(s',a,\mu') \fand Z(\mu,\mu').
\end{equation}
Since $Z(s,s') > 0$ and $E(s,a,\mu) = 1$, it follows that $E(s',a,\mu') > 0$. Hence, $\tuple{s',a,\mu'} \in \delta$ and $E(s',a,\mu') = 1$. By~\eqref{eq: JHFJH}, it follows that $Z(s,s') \leq Z(\mu,\mu')$. 
We now prove that $Z(\mu,\mu') \leq R^\ddag(\mu,\mu')$, which allows to derive $R(s,s') \leq R^\ddag(\mu,\mu')$. 
By~\eqref{eq: JHFKJ}, it suffices to prove that 
\begin{eqnarray}
\!\!\!\!\!\!\!\!\!\!\!\!\!&& \textrm{for every $t \in S$, there exists $t' \in S$ such that $Z(\mu,\mu') \leq (\mu(t) \fto R(t,t') \fand \mu'(t'))$} \label{eq: HJHSJ1}\\[1ex]
\!\!\!\!\!\!\!\!\!\!\!\!\!&& \textrm{for every $t' \in S$, there exists $t \in S$ such that $Z(\mu,\mu') \leq (\mu(t') \fto R(t,t') \fand \mu(t))$.} \label{eq: HJHSJ2}
\end{eqnarray}

Consider~\eqref{eq: HJHSJ1} and let $t \in S$. Without loss of generality, assume that $Z(\mu,\mu') > 0$ and $\mu(t) > 0$. 
Since $Z$ is a fuzzy bisimulation of $G$, there exists $t' \in V$ such that 
\begin{equation}\label{eq: JHDKA}
Z(\mu,\mu') \fand E(\mu, \varepsilon, t) \leq E(\mu', \varepsilon, t') \fand Z(t,t').
\end{equation}
Since $Z(\mu,\mu') > 0$ and $E(\mu, \varepsilon, t) = \mu(t) > 0$, we have $E(\mu', \varepsilon, t') > 0$, which implies $t' \in S$ and $Z(t,t') = R(t,t')$. 
Since $E(\mu, \varepsilon, t) = \mu(t)$ and $E(\mu', \varepsilon, t') = \mu'(t')$, it follows from~\eqref{eq: JHDKA} that  
\[ Z(\mu,\mu') \fand \mu(t) \leq \mu'(t') \fand R(t,t'), \]
which implies~\eqref{eq: HJHSJ1}. 
Analogously, it can be shown that \eqref{eq: HJHSJ2} also holds. 
Thus, we have proved that the condition~(a) of Definition~\eqref{def: HGLAK} holds. 
Similarly, it can be proved that the condition~(b) of Definition~\eqref{def: HGLAK} also holds. 
Therefore, $R$ is a fuzzy bisimulation of~$\mS$. 
\myend
\end{proof}

\begin{corollary}\label{cor: JHFKL}
Let $\mS = \tuple{S, \SA, \delta}$ be a finite \NFTS and $G = \tuple{V, E, L, \SV, \SE}$ the \FLG corresponding to $\mS$. If $Z$ is the greatest fuzzy bisimulation of $G$, then $R = Z|_{S \times S}$ is the greatest fuzzy bisimulation of~$\mS$.
\end{corollary}

This corollary follows from Theorem~\ref{theorem: JHFKJ} in the same way as Corollary~\ref{cor: KJDKW} follows from Theorem~\ref{theorem: HFLWA}.

%\section{Computing bisimulations between nondeterministic fuzzy transition systems}
\section{Computing the greatest crisp/fuzzy bisimulation of a finite \NFTS}
\label{section: computation}

We present Algorithm~\ref{algCompCPfNFTS} on page~\pageref{algCompCPfNFTS} (resp.\ Algorithm~\ref{algCompFPfNFTS} on page~\pageref{algCompFPfNFTS}) for computing the crisp (resp.\ compact fuzzy) partition that corresponds to the greatest crisp (resp.\ fuzzy) bisimulation of a given finite \NFTS. They are based on the results of the previous section and the algorithms given in \cite{DBLP:journals/ijar/NguyenT24,DBLP:journals/isci/Nguyen23}, which deal with computing bisimulations for \FLGs. 
We have implemented these algorithms in Python and made our implementation publicly available~\cite{NFTS-impl}.

\begin{algorithm}[t]
\caption{\CompCPfNFTS\label{algCompCPfNFTS}}
\Input{a finite \NFTS $\mS = \tuple{S, \SA, \delta}$.}
\Output{the partition corresponding to the greatest crisp bisimulation of~$\mS$.}

\BlankLine
construct the \FLG $G$ corresponding to $\mS$\label{step: algCompCPfNFTS 1}\;
execute the algorithm \CompCBt from \cite{DBLP:journals/ijar/NguyenT24} for $G$ to compute the partition $\bbP$ that corresponds to the greatest crisp bisimulation of~$G$\label{step: algCompCPfNFTS 2}\;
$\rs := \emptyset$\label{step: algCompCPfNFTS 3}\;
\ForEach{$B \in \bbP$}{
    let $x$ be any element of $B$\;
    \lIf{$x \in S$}{add $B$ to $\rs$\label{step: algCompCPfNFTS 6}}
}
\Return $\rs$\label{step: algCompCPfNFTS 7}\;
\end{algorithm}

\begin{example}
Consider the execution of Algorithm~\ref{algCompCPfNFTS} for the \NFTS $\mS$ given in Example~\ref{example: HGFJS}. 
The \FLG $G$ corresponding to $\mS$ has been specified in that example. 
Executing the algorithm \CompCBt from \cite{DBLP:journals/ijar/NguyenT24} for $G$ results in the partition 
$\bbP = \{\{s_1\}$, $\{s_2, s_5\}$, $\{s_3, s_4\}$, $\{\mu_1\}$, $\{\mu_2\}$, $\{\mu_3\}\}$. 
Executing the steps \ref{step: algCompCPfNFTS 3}--\ref{step: algCompCPfNFTS 7} of Algorithm~\ref{algCompCPfNFTS} results in the partition $\{\{s_1\}$, $\{s_2, s_5\}$, $\{s_3, s_4\}\}$. This can be checked by using our implementation~\cite{NFTS-impl}. When the implemented program is run with the option ``-{}-verbose'', it also displays $\bbP$ and information about intermediate steps of the algorithm \CompCBt.
\myend
\end{example}

\begin{theorem}
Algorithm~\ref{algCompCPfNFTS} is a correct algorithm for computing the partition corresponding to the greatest crisp bisimulation of a finite \NFTS $\mS = \tuple{S, \SA, \delta}$. It can be implemented to run in time of the order $O((\size(\delta) \log{l} + |S|) \log{(|S| + |\deltatt|)})$, where $l$ is the number of fuzzy values used in~$\mS$ plus~2. 
\end{theorem}

Note that $|\deltatt| \leq |\delta|$ and the occurrence of $|\deltatt|$ in the above complexity order can be replaced by $|\delta|$. Also note that, when $|\delta| \geq |S|$, that complexity order is within $O(\size(\delta) \cdot \log{|\delta|} \cdot \log{l})$, $O(|S| \cdot |\delta| \cdot \log{|\delta|} \cdot \log{l})$ and $O(|S| \cdot |\delta| \cdot \log^2{|\delta|})$.

\begin{proof}
Let $G = \tuple{V, E, L, \SV, \SE}$ and $\bbP$ be the objects mentioned in Algorithm~\ref{algCompCPfNFTS} and let $Z$ be the greatest crisp bisimulation of $G$. Thus, $\bbP$ is the partition corresponding to the equivalence relation~$Z$. By Corollary~\ref{cor: KJDKW}, $R = Z \cap (S \times S)$ is the greatest crisp bisimulation of~$\mS$. By the definition of $L$, if $x Z x'$, then $x, x' \in S$ or $x,x' \notin S$. Hence, a block $B \in \bbP$ belongs to the partition corresponding to the equivalence relation $R$ iff some elements of $B$ belong to $S$. Therefore, the set $\rs$ computed by the steps~\ref{step: algCompCPfNFTS 3}--\ref{step: algCompCPfNFTS 6} of Algorithm~\ref{algCompCPfNFTS} is really the partition corresponding to the greatest crisp bisimulation~$R$ of~$\mS$. 

The case $|\delta| = 0$ is trivial. So, assume that $|\delta| > 0$. 
By Remark~\ref{remark: JHFJS}, the step~\ref{step: algCompCPfNFTS 1} of Algorithm~\ref{algCompCPfNFTS} can be done in time of the order \InitComplexity. 
By~\cite[Theorem~4.2]{DBLP:journals/ijar/NguyenT24}, the step~\ref{step: algCompCPfNFTS 2} can be done in time of the order $O((m \log{l} + n) \log{n})$, where $n = |V| = |S| + |\deltatt|$ and $m = |\support(E)| = \size(\delta)$. 
The steps~\ref{step: algCompCPfNFTS 3}--\ref{step: algCompCPfNFTS 6} can be done in time of the order $O(n)$. 
Summing up, Algorithm~\ref{algCompCPfNFTS} can be implemented to run in time of the order $O((\size(\delta) \log{l} + |S|) \log{(|S| + |\deltatt|)})$. 
\myend
\end{proof}

Given $B$ as the compact fuzzy partition of a fuzzy equivalence relation, by $B.\anyElem()$ we denote any element of $B$. This method can be implemented as follows: if $B$ is a crisp block, then return any element of the set $B.\elements$; else let $B'$ be any element of the set $B.\subblocks$ and return $B'.\anyElem()$. 
Similarly, by $B.\allElems()$ we denote the (crisp) set of all elements of~$B$. This method is used in Algorithm~\ref{algCompFPfNFTS} and can be implemented as follows: if $B$ is a crisp block, then return $B.\elements$; else return the union of all the sets $B'.\allElems()$ with $B' \in B.\subblocks$. 

\begin{algorithm}[t]
\caption{\CompFPfNFTS\label{algCompFPfNFTS}}
\Input{a finite \NFTS $\mS = \tuple{S, \SA, \delta}$.}
\Output{the compact fuzzy partition corresponding to the greatest fuzzy bisimulation of~$\mS$ w.r.t.\ the G\"odel semantics.}

\BlankLine
construct the \FLG $G$ corresponding to $\mS$\label{step: algCompFPfNFTS 2}\;
execute the algorithm \CompFPt from \cite{DBLP:journals/isci/Nguyen23} for $G$ to compute the compact fuzzy partition $\bbB$ that corresponds to the greatest fuzzy bisimulation of~$G$ w.r.t.\ the G\"odel semantics\label{step: algCompFPfNFTS 3}\;
\lIf{$\delta = \emptyset$}{\Return $\bbB$\label{step: algCompFPfNFTS x}}
$P := \emptyset$\label{step: algCompFPfNFTS 4}\;
\ForEach{$B \in \bbB.\subblocks$}{
    \lIf{$B.\anyElem() \in S$}{add $B$ to the set $P$}
}
\lIf{$P$ contains only one element}{\Return that element}
\lElse{\Return the fuzzy block $B$ with $B.\degree = 0$ and $B.\subblocks = P$\label{step: algCompFPfNFTS 8}}
\end{algorithm}

\begin{example}
Consider the execution of Algorithm~\ref{algCompFPfNFTS} for the \NFTS $\mS$ given in Example~\ref{example: HGFJS}. 
The \FLG $G$ corresponding to $\mS$ has been specified in that example. 
Executing the algorithm \CompFPt from~\cite{DBLP:journals/isci/Nguyen23} for $G$ results in the compact fuzzy partition 
$\bbB = 
\{
  \{
    \{s_1\}_1, 
    \{s_2, s_5\}_1 
  \}_{0.4}$,
  $\{s_3, s_4\}_1$, 
  $\{
    \{
      \{\mu_1\}_1$, 
      $\{\mu_3\}_1
    \}_{0.5}$, 
    $\{\mu_2\}_1 
  \}_{0.4} 
\}_0$. 
Executing the steps \ref{step: algCompFPfNFTS 4}--\ref{step: algCompFPfNFTS 8} of Algorithm~\ref{algCompFPfNFTS} results in the compact fuzzy partition 
$\{
  \{
    \{s_1\}_1, 
    \{s_2, s_5\}_1 
  \}_{0.4}$,
  $\{s_3, s_4\}_1
\}_0$.  
This can be checked by using our implementation~\cite{NFTS-impl}. When the implemented program is run with the option ``-{}-verbose'', it also displays $\bbB$ and information about intermediate steps of the algorithm \CompFPt.
The fuzzy equivalence relation corresponding to the resultant compact fuzzy partition is given below. 
\[
\begin{array}{|c||c|c|c|c|c|}
\hline
 & s_1 & s_2 & s_3 & s_4 & s_5 \\
\hline\hline
s_1 & 1 & 0.4 & 0 & 0 & 0.4 \\
\hline
s_2 & 0.4 & 1 & 0 & 0 & 1 \\
\hline
s_3 & 0 & 0 & 1 & 1 & 0 \\
\hline
s_4 & 0 & 0 & 1 & 1 & 0 \\
\hline
s_5 & 0.4 & 1 & 0 & 0 & 1 \\
\hline
\end{array}
\]
It is the greatest fuzzy bisimulation of~$\mS$ w.r.t.\ the G\"odel semantics. 
\myend
\end{example}

\begin{theorem}
Algorithm~\ref{algCompFPfNFTS} is a correct algorithm for computing the compact fuzzy partition corresponding to the greatest fuzzy bisimulation of a finite \NFTS $\mS = \tuple{S, \SA, \delta}$ w.r.t.\ the G\"odel semantics. It can be implemented to run in time of the order $O((\size(\delta) \log{l} + |S|) \log{(|S| + |\deltatt|)})$, where $l$ is the number of fuzzy values used in~$\mS$ plus~2. 
\end{theorem}

As stated for Algorithm~\ref{algCompCPfNFTS}, the occurrence of $|\deltatt|$ in the above complexity order can be replaced by $|\delta|$. Also note that, when $|\delta| \geq |S|$, that complexity order is within $O(\size(\delta) \cdot \log{|\delta|} \cdot \log{l})$, $O(|S| \cdot |\delta| \cdot \log{|\delta|} \cdot \log{l})$ and $O(|S| \cdot |\delta| \cdot \log^2{|\delta|})$.

\begin{proof}
For the theorem and this proof, $\fand$ is assumed to be the G\"odel t-norm. 
Let $G = \tuple{V, E, L, \SV, \SE}$ and $\bbB$ be the objects mentioned in Algorithm~\ref{algCompFPfNFTS} and let $Z$ be the greatest fuzzy bisimulation of~$G$. Thus, $\bbB$ is the compact fuzzy partition corresponding to the fuzzy equivalence relation~$Z$. By Corollary~\ref{cor: JHFKL}, $Z|_{S \times S}$ is the greatest fuzzy bisimulation of~$\mS$.
The case $\delta = \emptyset$ is clear. So, assume that $\delta \neq \emptyset$. 
By the definition of $L$, for $x,x' \in V$, if $Z(x, x') > 0$, then $x, x' \in S$ or $x,x' \notin S$. 
Since $V = S \cup \deltatt$ and $\delta \neq \emptyset$, $\bbB$ must be a fuzzy block with $\bbB.\degree = 0$, and for any $B \in \bbB.\subblocks$, either $B.\allElems() \subseteq S$ or $B.\allElems() \subseteq \deltatt$. 
If $B$ is a unique block from $\bbB.\subblocks$ with $B.\anyElem() \in S$, then $B$ is the compact fuzzy partition corresponding to the fuzzy equivalence relation~$Z|_{S \times S}$. If $B_1, \ldots, B_k$ are all the blocks from $\bbB.\subblocks$ with $B_i.\anyElem() \in S$, for $1 \leq i \leq k$, and $k > 1$, then $\{B_1, \ldots, B_k\}_0$ is the compact fuzzy partition corresponding to the fuzzy equivalence relation~$Z|_{S \times S}$. Hence, by the steps~\ref{step: algCompFPfNFTS 4}--\ref{step: algCompFPfNFTS 8}, Algorithm~\ref{algCompFPfNFTS} returns the compact fuzzy partition corresponding to the greatest fuzzy bisimulation of~$\mS$.

By Remark~\ref{remark: JHFJS}, the step~\ref{step: algCompFPfNFTS 2} can be done in time of the order \InitComplexity. 
By~\cite[Theorem~4.12]{DBLP:journals/isci/Nguyen23}, the step~\ref{step: algCompFPfNFTS 3} can be done in time of the order $O((m \log{l} + n) \log{n})$, where $n = |V| = |S| + |\deltatt|$ and $m = |\support(E)| = \size(\delta)$. 
The step~\ref{step:
 algCompFPfNFTS x} of Algorithm~\ref{algCompFPfNFTS} runs in constant time. 
The steps~\ref{step: algCompFPfNFTS 4}--\ref{step: algCompFPfNFTS 8} run in time of the order $O(|V|) = O(|S| + |\deltatt|)$. 
Summing up, Algorithm~\ref{algCompFPfNFTS} can be implemented to run in time of the order $O((\size(\delta) \log{l} + |S|) \log{(|S| + |\deltatt|)})$. 
\myend
\end{proof}

The work~\cite{DBLP:journals/isci/Nguyen23} provides a function named {\em ConvertFP2FB} for converting a compact fuzzy partition of a finite set $S$ to the corresponding fuzzy equivalence relation (when $\fand$ is the G\"odel t-norm). Its time complexity is of order $O(|S|^2)$. We do not need to explicitly keep the greatest fuzzy bisimulation $Z$ of a finite \NFTS $\mS = \tuple{S, \SA, \delta}$ with that cost. A compact fuzzy partition is implemented in~\cite{DBLP:journals/isci/Nguyen23} as a tree, where each node has a reference to its parent. Given $x, y \in S$ and the compact fuzzy partition $B$ returned by Algorithm~\ref{algCompFPfNFTS} for $\mS$, computing $Z(x,y)$ is reduced to the task of finding the lowest common ancestor of the leaves of the tree representing $B$ that contain $x$ and $y$, respectively. This latter task can be done efficiently by using the algorithm of Harel and Tarjan~\cite{DBLP:journals/siamcomp/HarelT84}. 

%===============================================================================
%\section{Simulations and bisimulations between nondeterministic fuzzy labeled transition systems}
\section{Extending \NFTSs with fuzzy state labels}
\label{section: extension}

We define a {\em nondeterministic fuzzy labeled transition system} (\NFLTS) as an extension of an \NFTS in which each state is labeled by a fuzzy subset of an alphabet~$\Sigma$. In particular, an \NFLTS is a structure \mbox{$\mS = \tuple{S, \SA, \delta, \sS, L}$}, where $S$, $\SA$ and $\delta$ are as for an \NFTS, $\sS$ is a set of state labels, and $L: S \to \mF(\sS)$ is the state labeling function. It is {\em finite} if all the components $S$, $\SA$, $\delta$ and $\sS$ are finite. 

In this section, we first define the notions of a crisp/fuzzy auto-bisimulation of an \NFLTS and prove that Algorithms~\ref{algCompCPfNFTS} and~\ref{algCompFPfNFTS} are still correct when taking a finite \NFLTS as the input instead of a finite \NFTS. We then define four notions of a crisp/fuzzy simulation/bisimulation between two \NFLTSs and state what existing results on logical and algorithmic characterizations of simulations/bisimulations for fuzzy structures of other kinds can be reformulated for \NFLTSs. 
In particular, we present efficient algorithms for computing the greatest crisp (resp.\ fuzzy) simulation between two finite \NFLTSs (under the G\"odel semantics in the case of fuzzy simulation). 

We proceed by extending the notion of the corresponding \FLG for \NFLTSs appropriately, preserving the state labeling function. In particular, the definition given below differs from Definition~\ref{def: JHFSO 1} only in the specification of~$\SV$ and~$L$. 

\begin{definition}\label{def: JHFSO 2}
Given an \NFLTS $\mS = \tuple{S, \SA, \delta, \sS, L_0}$, the {\em \FLG corresponding to $\mS$} is the \FLG $G = \tuple{V, E, L, \SV, \SE}$ specified as follows:\footnote{In this definition, $L_0$ is the state labeling function of~$\mS$, whereas $L$ is the vertex labeling function of~$G$.}
\begin{itemize}
\item $V$, $\SE$ and $E$ are as in Definition~\ref{def: JHFSO 1};
\item $\SV = \sS \cup \{s\}$, where $s \notin \sS$ stands for ``being a state''; 
\item $L: V \to \mF(\SV)$ is specified by: 
    \begin{itemize}
    \item $L(x)|_{\Sigma} = L_0(x)$ and $L(x)(s) = 1$ for $x \in S$, 
    \item $\support(L(x)) = \emptyset$ for $x \in \deltatt$.
\myend
    \end{itemize}
\end{itemize}
\end{definition}

In the spirit of Theorems~\ref{theorem: HFLWA} and~\ref{theorem: JHFKJ}, we define bisimulations for \NFLTSs as follows. 

\begin{definition}
Let $\mS = \tuple{S, \SA, \delta, \sS, L}$ be an \NFLTS and $G$ its corresponding \FLG. 
A relation $R \subseteq S \times S$ is called a {\em crisp bisimulation} of $\mS$ if there exists a crisp bisimulation $Z$ of $G$ such that $R = Z \cap (S \times S)$. A fuzzy relation $R \in \mF(S \times S)$ is called a {\em fuzzy bisimulation} of $\mS$ if there exists a fuzzy bisimulation $Z$ of $G$ such that $R = Z|_{S \times S}$.  
\myend
\end{definition}

The following result is a consequence of this definition. 

\begin{proposition}\label{prop: HGRJK}
Taking a finite \NFLTS $\mS$ as the input instead of a finite \NFTS, 
Algorithm~\ref{algCompCPfNFTS} is a correct algorithm for computing the partition corresponding to the greatest crisp bisimulation of~$\mS$, and 
Algorithm~\ref{algCompFPfNFTS} is a correct algorithm for computing the compact fuzzy partition corresponding to the greatest fuzzy bisimulation of~$\mS$ under the G\"odel semantics. 
\end{proposition}

\begin{proof}
Let $\mS = \tuple{S, \SA, \delta, \sS, L_0}$ and let $G = \tuple{V, E, L, \SV, \SE}$ be the \FLG corresponding to $\mS$. 

Consider the case of Algorithm~\ref{algCompCPfNFTS}. Let $\bbP$ be the object mentioned in Algorithm~\ref{algCompCPfNFTS} and $Z$ the greatest crisp bisimulation of $G$. Thus, $\bbP$ is the partition corresponding to the equivalence relation~$Z$. Due to the use of $s \in \SV$ and by the definition of $L$, if $x Z x'$, then $x, x' \in S$ or $x,x' \notin S$. Hence, the set $\rs$ computed by the steps~\ref{step: algCompCPfNFTS 3}--\ref{step: algCompCPfNFTS 6} of Algorithm~\ref{algCompCPfNFTS} is the partition corresponding to the equivalence relation~$Z \cap (S \times S)$, which is the greatest crisp bisimulation of~$\mS$ (by definition). 

Consider the case of Algorithm~\ref{algCompFPfNFTS}. Let $\bbB$ be the object mentioned in Algorithm~\ref{algCompFPfNFTS} and $Z$ the greatest fuzzy bisimulation of~$G$ w.r.t.\ the G\"odel semantics. Thus, $\bbB$ is the compact fuzzy partition corresponding to the fuzzy equivalence relation~$Z$. 
Due to the use of $s \in \SV$ and by the definition of $L$, if $Z(x, x') > 0$, then $x, x' \in S$ or $x,x' \notin S$. 
Hence, by the steps~\ref{step: algCompFPfNFTS x}--\ref{step: algCompFPfNFTS 8}, Algorithm~\ref{algCompFPfNFTS} returns the compact fuzzy partition corresponding to the fuzzy equivalence relation~$Z|_{S \times S}$, which is the greatest fuzzy bisimulation of~$\mS$ (by definition).
\myend
\end{proof}

\begin{definition}\label{def: HGDKA}
Let $G = \tuple{V, E, L, \SV, \SE}$ and $G' = \tuple{V', E', L', \SV, \SE}$ be \FLGs (over the same signature $\tuple{\SV,\SE}$). 
A relation $Z \subseteq V \times V'$ is called a {\em crisp simulation} between $G$ and $G'$ if the following conditions hold for every $\tuple{x,x'} \in Z$ and $r \in \SE$: 
\begin{itemize}
\item $L(x) \leq L(x')$, 
\item for every $y \in V$ with \mbox{$E(x,r,y) > 0$}, there exists $y'\in V'$ such that $y Z y'$ and \mbox{$E(x,r,y) \leq E(x',r,y')$}. 
\end{itemize}
A relation $Z \subseteq V \times V'$ is called a {\em crisp bisimulation} between $G$ and $G'$ if: $Z$ is a crisp simulation between $G$ and $G'$, and $Z^{-1}$ is a crisp simulation between $G'$ and $G$.
\myend
\end{definition}

The above definition is consistent with Definition~\ref{def: DSHGQ} when $Z \neq \emptyset$. That is, a non-empty relation $Z$ is a crisp bisimulation of $G$ iff it is a crisp bisimulation between $G$ and itself. The condition on non-emptiness is just a technical matter: there always exists a (non-empty) crisp bisimulation of a \FLG $G$, but it is possible that there is only one crisp bisimulation between \FLGs $G$ and $G'$ and it is the empty relation. 
In general, Definition~\ref{def: DSHGQ} can be loosened by discarding the condition on non-emptiness.

\begin{definition}\label{def: JRKAF}
Let $G = \tuple{V, E, L, \SV, \SE}$ and $G' = \tuple{V', E', L', \SV, \SE}$ be \FLGs (over the same signature $\tuple{\SV,\SE}$). 
A fuzzy relation $Z \in \mF(V \times V')$ is called a {\em fuzzy simulation} between $G$ and $G'$ (w.r.t.~$\fand$) if the following conditions hold for every $x,y \in V$, $x' \in V'$, $p \in \SV$ and $r \in \SE$:
\begin{itemize}
\item $Z(x,x') \leq (L(x)(p) \fto L(x')(p))$ %\label{eq: FS1}
\item $\E y' \in V'\ (Z(x,x') \fand E(x,r,y) \leq E(x',r,y') \fand Z(y,y')).$ %\label{eq: FS2}
\end{itemize}
A fuzzy relation $Z \in \mF(V \times V')$ is called a {\em fuzzy bisimulation} between $G$ and $G'$ if: $Z$ is a fuzzy simulation between $G$ and $G'$, and $Z^{-1}$ is a fuzzy simulation between $G'$ and $G$. 
\myend
\end{definition}

The above definition is consistent with Definition~\ref{def: HGAKF}. That is, a fuzzy relation $Z$ is a fuzzy bisimulation of $G$ iff it is a fuzzy bisimulation between $G$ and itself. 

In the spirit of Theorems~\ref{theorem: HFLWA} and~\ref{theorem: JHFKJ}, we define crisp/fuzzy simulations/bisimulations between \NFLTSs as follows. 

\begin{definition}\label{def: JHFKX}
Let $\mS = \tuple{S, \SA, \delta, \sS, L}$ and $\mS' = \tuple{S', \SA, \delta', \sS, L'}$ be \NFLTSs (over the same signature $\tuple{A,\sS}$). Let $G$ and $G'$ be the \FLGs corresponding to $\mS$ and $\mS'$, respectively. 
A relation $R \subseteq S \times S'$ is called a {\em crisp simulation} (resp.\ {\em crisp bisimulation}) between $\mS$ and $\mS'$ if there exists a crisp simulation (resp.\ crisp bisimulation) $Z$ between $G$ and $G'$ such that $R = Z \cap (S \times S')$. 
A fuzzy relation $R \in \mF(S \times S')$ is called a {\em fuzzy simulation} (resp.\ {\em fuzzy bisimulation}) between $\mS$ and $\mS'$ if there exists a fuzzy simulation (resp.\ fuzzy bisimulation) $Z$ between $G$ and $G'$ such that $R = Z|_{S \times S'}$.  
\myend
\end{definition}

Note that our notion of a crisp (resp.\ fuzzy) simulation when restricted (from \NFLTSs) to \NFTSs is different in nature from the one defined in~\cite{DBLP:journals/fss/WuCBD18} (resp.\ \cite{DBLP:journals/tfs/QiaoZF23}). In particular, our notions of a crisp/fuzzy simulation take into account only the ``forward'' direction, while the notions of a crisp/fuzzy simulation defined in \cite{DBLP:journals/fss/WuCBD18,DBLP:journals/tfs/QiaoZF23} take into account a mixture of the ``forward'' direction for the distribution level and both the ``forward'' and ``backward'' directions for the lifting level (expressed by~\eqref{eq: lifted relation} and \eqref{eq: JHFKJ}). The former ones relate to the preservation of the existential fragments of modal logics. In addition, the use of ``$\leq$'' instead of ``$=$'' in the condition~(a) of Definition~\ref{def: HGDKA} and the use of ``$\fto$'' instead of ``$\fequiv$'' in the condition~(a) of Definition~\ref{def: JRKAF} relate to the preservation of the positive fragments of modal logics. Together, our notions of crisp/fuzzy simulations relate to the preservation of the positive existential fragments of modal logics~\cite{BRV2001}. 

Each \FLG can be treated as a fuzzy Kripke model, a fuzzy interpretation in description logic or a fuzzy labeled transition system (\FLTS). In accordance with Definition~\ref{def: JHFKX}, known results on logical characterizations of crisp/fuzzy bisimulations/simulations in fuzzy modal/description logics or between \FLTSs can be applied to \NFLTSs. Notable are the following. 
\begin{itemize}
\item The logical characterizations of crisp bisimulations that are formulated and proved for fuzzy description logics in~\cite{DBLP:journals/tfs/NguyenN23} can be restated for \NFLTSs by defining semantics of concepts directly using an \NFLTS instead of the corresponding \FLG treated as an interpretation in description logic. 

\item The logical characterizations of fuzzy bisimulations (respectively, fuzzy simulations) that are formulated and proved for fuzzy modal logics in~\cite{FBSML} (respectively, \cite{DBLP:journals/cas/NguyenN22}) can be restated for \NFLTSs by defining semantics of modal formulas directly using an \NFLTS instead of the corresponding \FLG treated as a Kripke model. 

\item The logical characterizations of crisp simulations that are formulated and proved for \FLTSs in~\cite{DBLP:conf/fuzzIEEE/NguyenN21} can be restated for \NFLTSs by defining semantics of modal formulas directly using an \NFLTS instead of the corresponding \FLG treated as an \FLTS.
\end{itemize}

Clearly, one can also extend the logical characterizations of crisp (respectively, fuzzy) bisimulations formulated for \NFTSs in~\cite{DBLP:journals/fss/WuCBD18} (respectively, \cite{DBLP:journals/tfs/QiaoZF23}) to deal with \NFLTSs.

Computation of the greatest crisp/fuzzy bisimulation between two finite \NFLTSs $\mS$ and $\mS'$ (under the G\"odel semantics in the case of fuzzy bisimulation) can be reduced to the task of computing the greatest crisp/fuzzy bisimulation of the \NFLTS being the disjoint union of $\mS$ and $\mS'$, in the way stated in \cite[Section~5]{DBLP:journals/isci/Nguyen23} and using Algorithms~\ref{algCompCPfNFTS} and~\ref{algCompFPfNFTS} for \NFLTSs as stated in Proposition~\ref{prop: HGRJK}. 
Once again, we do not need to explicitly transform the resultant crisp (resp.\ compact fuzzy) partition to the corresponding crisp (resp.\ fuzzy) bisimulation, but can use the algorithm of Harel and Tarjan~\cite{DBLP:journals/siamcomp/HarelT84} instead.

The algorithm {\em ComputeSimulationEfficiently} provided in~\cite{DBLP:journals/jifs/Nguyen22} for computing the greatest crisp simulation between two finite \FLTSs can be used to produce an efficient algorithm for computing the greatest crisp simulation between two finite \NFLTSs as follows. 

\medskip

\begin{algorithm}[H]
\caption{$\mathsf{ComputeCrispSimulationNFLTS}$\label{algCompCS-NFLTS}}
\Input{finite \NFLTSs $\mS$ and $\mS'$.}
\Output{the greatest crisp simulation between $\mS$ and $\mS'$.}
\BlankLine
\label{step: algCompCS-NFLTS 1} construct the \FLGs $G$ and $G'$ that correspond to~$\mS$ and~$\mS'$, respectively\; 
\label{step: algCompCS-NFLTS 2} treating these \FLGs as \FLTSs (in the usual way), apply the algorithm {\em ComputeSimulationEfficiently} given in~\cite{DBLP:journals/jifs/Nguyen22} to compute the greatest crisp simulation $Z$ between $G$ and $G'$\; 
\label{step: algCompCS-NFLTS 3} \Return $Z \cap (S \times S')$\;
\end{algorithm}

\medskip

The algorithm {\em ComputeFuzzySimulation} provided in~\cite{TFS2020} for computing the greatest fuzzy simulation between two finite fuzzy interpretations in the fuzzy description logic \fALC under the G\"odel semantics can be used to produce an efficient algorithm for computing the greatest fuzzy simulation between two finite \NFLTSs as follows for the case where $\fand$ is the G\"odel t-norm.

\medskip

\begin{algorithm}[H]
\caption{$\mathsf{ComputeFuzzySimulationNFLTS}$\label{algCompFS-NFLTS}}
\Input{finite \NFLTSs $\mS$ and $\mS'$.}
\Output{the greatest fuzzy simulation between $\mS$ and $\mS'$ w.r.t.\ the G\"odel semantics.}
\BlankLine
\label{step: algCompFS-NFLTS 1} construct the \FLGs $G$ and $G'$ that correspond to~$\mS$ and~$\mS'$, respectively\; 
\label{step: algCompFS-NFLTS 2} treating these \FLGs as interpretations in description logic (in the usual way), apply the algorithm {\em ComputeFuzzySimulation} given in~\cite{TFS2020} to compute the greatest fuzzy simulation between $G$ and $G'$ (in \fALC) under the G\"odel semantics\;
\label{step: algCompFS-NFLTS 3} \Return $Z|_{S \times S'}$\;
\end{algorithm}

\medskip

\begin{theorem}
Algorithm~\ref{algCompCS-NFLTS} (resp.\ \ref{algCompFS-NFLTS}) is a correct algorithm for computing the greatest crisp (resp.\ fuzzy) simulation between finite \NFLTSs $\mS = \tuple{S, \SA, \delta, \sS, L}$ and $\mS' = \tuple{S', \SA, \delta', \sS, L'}$. Its time complexity is of the order $O((m+n)n)$, where $m = \size(\delta) + \size(\delta')$ and $n = |S| + |S'| + |\deltatt| + |\deltatt'|$, treating $|\SA|$ and $|\sS|$ as constants.
\end{theorem}

\begin{proof}
The correctness of Algorithm~\ref{algCompCS-NFLTS} (resp.\ \ref{algCompFS-NFLTS}) directly follows from Definition~\ref{def: JHFKX} and the correctness of the algorithm {\em ComputeSimulationEfficiently} given in~\cite{DBLP:journals/jifs/Nguyen22} (resp.\ {\em ComputeFuzzySimulation} given in~\cite{TFS2020}). 
By Remark~\ref{remark: JHFJS}, the step~\ref{step: algCompCS-NFLTS 1} can be done in time of the order $O(m+n)$. 
By~\cite[Theorem~3.5]{DBLP:journals/jifs/Nguyen22} (resp.\ \cite[Theorem~20]{TFS2020}), the step~\ref{step: algCompCS-NFLTS 2} runs in time of the order $O((m+n)n)$. 
The step~\ref{step: algCompCS-NFLTS 3} runs in time of the order $O(n^2)$. 
Hence, Algorithm~\ref{algCompCS-NFLTS} (resp.\ \ref{algCompFS-NFLTS}) runs in time of the order $O((m+n)n)$. 
\myend
\end{proof}

%===============================================================================

\section{Conclusions}
\label{section: conc}

We have provided efficient algorithms for computing the partition corresponding to the greatest crisp bisimulation of a finite \NFLTS $\mS = \tuple{S, \SA, \delta, \sS, L}$, as well as the compact fuzzy partition corresponding to the greatest fuzzy bisimulation of $\mS$ under the G\"odel semantics. Their time complexities are of the order $O((\size(\delta) \log{l} + |S|) \log{(|S| + |\deltatt|)})$, where $l$ is the number of fuzzy values used in~$\mS$ plus~2. If needed, one can explicitly convert  a crisp (resp.\ compact fuzzy) partition to the corresponding crisp (resp.\ fuzzy) equivalence relation in time of the order $O(|S|^2)$. However, the conversion can be avoided by exploiting the algorithm of finding the lowest common ancestor by Harel and Tarjan~\cite{DBLP:journals/siamcomp/HarelT84}. 
Our algorithms when used for computing the greatest crisp/fuzzy bisimulation of a finite \NFTS significantly outperform the previously known algorithms \cite{DBLP:journals/fss/WuCBD18,DBLP:journals/tfs/QiaoZF23} for the task, like comparing $O(|S| \cdot |\delta| \cdot \log^2{|\delta|})$ with $O(|S|^4 \cdot |\delta|^2)$ and $O(|S|^4 \cdot |\delta|^2 \cdot l)$. 

We have also provided efficient algorithms for computing the greatest crisp/fuzzy simulation between two finite \NFLTSs. 

%===============================================================================

\biboptions{sort&compress}
\bibliography{BSfDL}
\bibliographystyle{plain}

%-------------------------------------------------------------------------------

\end{document}